\documentclass[envcountsame]{llncs}
\usepackage{amssymb,amsmath}
\usepackage{xspace}
\usepackage{tikz}
\usepackage[draft]{fixme}
\usepackage{color}
\usepackage{subfig}

\usepackage[utf8]{inputenc}
\usepackage{amsfonts}
\usepackage{mathtools}
\usepackage{bbm}
\usepackage[capitalize]{cleveref}
\usepackage{enumerate}
\usepackage{mathrsfs}

\usetikzlibrary{positioning}
\usetikzlibrary{automata}

\spnewtheorem*{proposition*}{Proposition}{\bfseries}{\itshape}
\spnewtheorem*{lemma*}{Lemma}{\bfseries}{\itshape}
\spnewtheorem*{theorem*}{Theorem}{\bfseries}{\itshape}
\spnewtheorem*{corollary*}{Corollary}{\bfseries}{\itshape}
\spnewtheorem*{remark*}{Remark}{\bfseries}{\itshape}

\newcommand{\aut}[1]{\mathfrak{#1}}

\newcommand{\naturals}{\mathbbm{N}}

\newcommand{\rec}{\mathsf{Rec}}
\newcommand{\occset}{\mathrm{occ}}
\newcommand{\infset}{\mathrm{inf}}
\newcommand{\alphinf}{\mathrm{alphinf}}
\newcommand{\alphabet}{\mathrm{alph}}
\newcommand{\ext}{\mathrm{ext}}

\newcommand{\REG}{\mathsf{REG}}
\newcommand{\BC}{\mathrm{BC}}

\newcommand{\F}{\mathcal{F}}
\newcommand{\K}{\mathcal{K}}

\newcommand{\set}[1]{\{#1\}}
\newcommand{\omReg}{$\omega$-regular\xspace}
\newcommand{\mStar}{\mathbb{M}(\Sigma, I)}

\newcommand{\mOmega}{\mathbb{R}(\Sigma, I)}

\newcommand{\trace}[1]{\sim_{#1}}

\newcommand{\prefix}{\sqsubseteq}
\newcommand{\pprefix}{\sqsubset}

\newcommand{\mapping}[3]{#1\colon #2\rightarrow #3}
\newcommand{\bigO}{\boldsymbol{\mathcal{O}}}
\newcommand{\Pow}{\mathscr{P}}

\title{Weak $\omega$-Regular Trace Languages}
\author{Namit Chaturvedi \and Marcus Gelderie}
\institute{RWTH Aachen University, Lehrstuhl f{\"{u}}r Informatik 7, D-52056 Aachen
\email{\{chaturvedi,gelderie\}@automata.rwth-aachen.de}}

\begin{document}

\maketitle

\begin{abstract}
Mazurkiewicz traces describe concurrent behaviors of distributed systems. Trace-closed word languages, which are ``linearizations'' of trace languages, constitute a weaker notion of concurrency but still give us tools to investigate the latter. In this vein, our contribution is twofold. Firstly, we develop definitions that allow classification of \omReg trace languages in terms of the corresponding trace-closed \omReg word languages, capturing E-recognizable (reachability) and (deterministically) Büchi recognizable languages. Secondly, we demonstrate the first automata-theoretic result that shows the equivalence of \omReg trace-closed word languages and Boolean combinations of deterministically $I$-diamond B{\"u}chi recognizable trace-closed languages.
\end{abstract}

\section{Introduction}\label{intro}
Traces were introduced as models representing partially concurrent behaviors of distributed systems by Mazurkiewicz, who later also provided explicit definition of infinite traces \cite{Maz87}. Zielonka demonstrated the close relation between traces and words that can be viewed as ``linearizations'' of traces, and also established automata-theoretic results regarding recognizability of languages of finite traces \cite{Zie87} (alternatively, see \cite{Muk12} for an introduction). We also refer the reader to \cite{DR95} for a comprehensive collection of early results. Subsequently, Gastin-Petit \cite{GP92} and Diekert-Muscholl \cite{DM93}, respectively, demonstrated the direct correspondence between the family of recognizable languages of infinite traces (\omReg trace languages), and the families of asynchronous Büchi and deterministic asynchronous Muller automata. As with languages of finite traces, a set of infinite traces is recognizable iff the set of linearizations, i.e. the word language, corresponding to the set of infinite traces is. 

It is well known that \omReg languages can be obtained by various operations from regular languages of finite words. In general, any \omReg language $L$ can be represented as $K_1\cdot K_2^\omega$, with $K_1, K_2$ regular. Languages $L$ of this form are recognized by Muller automata. There are also notions of subclasses of \omReg languages that are obtained from given regular languages $K$ in the following ways:
\begin{itemize}
 \item $\ext(K) = \{\alpha \in \Sigma^\omega \mid \alpha \mbox{ has a prefix in } K\}$
 \item $\lim(K) = \{\alpha \in \Sigma^\omega \mid \alpha \mbox{ has infinitely many prefixes in } K\}$
\end{itemize}

For $K$ regular, languages $\lim(K)$ are referred to as deterministically Büchi recognizable languages, and the corresponding deterministic Büchi automata (DBAs) can be constructed efficiently from the minimal DFA recognizing $K$. The same is true for languages $\ext(K)$, which are recognized by $E$-automata (reachability automata). Finite Boolean combinations of languages $\ext(K)$ yield the family of \emph{weakly recognizable} languages. This class can alternatively be characterized in terms of automata, being precisely the class of languages recognizable by deterministic weak automata (DWAs). Finite Boolean combinations of languages $\lim(K)$ result in all \omReg languages. For a class $\K$ of regular languages, we refer to classes $\ext(\K), \lim(\K)$.

For both of these operations,  we define corresponding operations for recognizable languages $T$ of finite traces, $\ext(T)$ and $\lim(T)$. We show these operations relate to the classical word operations on the language $K$ of linearizations of traces in $T$. More precisely, given a language of finite traces $T$ with $K$ the language of its linearizations, we show how $K$ can be modified to a trace-closed $K_I$, such that the diagram in Fig. \ref{fig:ext_commutes} commutes. In particular, for every trace-closed $K$, $\ext(K_I)$ is trace-closed. Furthermore, for every recognizable $T$, the linearizations of $\ext(T)$ are recognizable by an $I$-diamond E-automaton. Using this, we characterize the class of languages of infinite traces whose linearizations are recognizable by $I$-diamond DWAs, as precisely the Boolean combinations of languages of the form $\ext(T)$ for recognizable languages $T$ of finite traces. In the same spirit, we consider $\lim(T)$ and $\lim(K)$. Here the situation is different, in that not for every recognizable $T$, the language of linearizations of $\lim(T)$ is recognizable by an $I$-diamond DBA. We characterize the subclass of recognizable $T$, where $\aut{A}_K$, the minimal DFA for the linearizations $K$, also recognizes the linearizations of $\lim(T)$ as a DBA. For those languages, the diagram \ref{fig:lim_commutes} commutes. In particular, for such $K$, $\lim(K)$ is trace-closed. Moreover, we show that every recognizable language of infinite traces is a finite Boolean combination of languages $\lim(T)$ for such $T$. Hence, any trace-closed language $L$ of infinite traces is a Boolean combination of $I$-diamond DBA recognizable trace-closed languages. 

\begin{figure}
\centering

\subfloat[Infinitary extensions]{\label{fig:ext_commutes}
\begin{tikzpicture}[every node/.style={inner sep =1mm, minimum width=0mm}]
\node (T) at (0,0) {$T$};
\node (K) at (0,-1) {$K$};
\node (extT) at (3,0) {$\ext(T)$};
\node (KI) at (1.3,-1) {$K_I$};
\node (extKI) at (3,-1) {$\ext(K)$};

\path[-latex] 
(T) edge (K)
    edge (extT)
(K) edge (KI)
(KI) edge (extKI)
(extT) edge (extKI);
\end{tikzpicture}
}
\hspace{.8cm}
\subfloat[Infinitary limits]{\label{fig:lim_commutes}
\begin{tikzpicture}[every node/.style={inner sep =1mm, minimum width=0mm}]
\node (T) at (0,0) {$T$};
\node (K) at (0,-1) {$K$};
\node (extT) at (3,0) {$\lim(T)$};
\node (extKI) at (3,-1) {$\lim(K)$};

\path[-latex] 
(T) edge (K)
    edge (extT)
(K) edge (extKI)
(extT) edge (extKI);
\end{tikzpicture}
}
\caption{Infinite trace-closed languages from finite trace-closed languages.}
\end{figure}
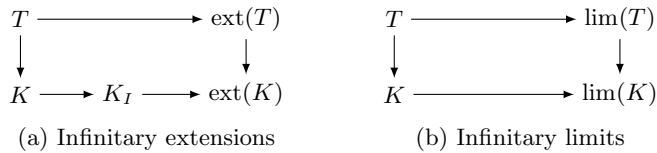

In related work, Muscholl and Diekert \cite{DM93} consider a form of ``deterministic'' trace languages. In \cite{Mus94} it is shown that every recognizable language of infinite traces is a Boolean combination of these deterministic languages. However, those languages require modifications to the B\"uchi acceptance condition in order to obtain a correspondence in terms of $I$-diamond DBAs. The problem of finding a suitable class of languages which has a classical automaton correspondence is left open in \cite{Mus94}.

We begin with presenting definitions that are relevant to the connections between regular and \omReg languages. We also formally introduce the notion of regular and \omReg trace languages. In Sec. \ref{sec:tcLanguages}, we present definitions that allow construction of various classes of \omReg trace languages from regular trace languages. In particular, we classify trace languages whose linearizations are weakly recognizable, and those whose linearizations are DBA recognizable. We establish that every \omReg trace language is a Boolean combination of those trace languages whose linearizations are DBA recognizable.

\section{Preliminaries}\label{prelim}

We denote a recognizable language of finite words, or simply a \emph{regular language}, with the upper case letter $K$ and a class of such languages with $\K$. Finite words are denoted with lower case letters $u$, $v$, $w$ etc. Infinite words are denoted by lower case Greek letters $\alpha$ and $\beta$, and a recognizable language of infinite words, or simply an \emph{\omReg language}, by upper case $L$. For a word $u$ or $\alpha$, we denote its infix starting at position $i$ and ending at position $j$ by $u[i,j]$ or $\alpha[i,j]$, and the $i^{th}$ letter with $u[i]$ or $\alpha[i]$. For a language $K$, we denote the complement language by $\overline{K}$.

We assume the reader is familiar with the notions of Deterministic Finite Automata (DFAs) and Deterministic Büchi Automata (DBAs). We say that a language is \emph{DBA recognizable} iff it is recognized by a DBA. For the class $\REG$ of regular languages, the class $\lim(\REG)$ coincides with the DBA recognizable languages. Further, the class $\BC(\lim(\REG))$ of finite Boolean combinations of languages from $\lim(\REG)$ is also the class of \omReg languages, and it coincides with the class of languages recognized by nondeterministic B\"uchi or deterministic Muller automata.

Recall that a Deterministic Weak Automaton (DWA) is a DBA where every strongly connected component of the transition graph has only accepting states or only rejecting states. For a regular language $K$, the minimal DFA recognizing $K$ also recognizes $\lim(K)$ as a DBA. Given the minimal DFA $\aut{A} = (Q, \Sigma, q_0, \delta, F)$ recognizing $K$, a DWA $\aut{A}' \coloneqq (Q', \Sigma, q_0, \delta', F')$ recognizing $\ext(K)$, respectively $\overline{\ext(K)}$, can be constructed as follows:
\begin{enumerate}
 \item For a symbol $\bot \notin Q$ and define $Q' \coloneqq (Q \setminus F) \cup \{\bot\}$.
 \item For each $q \in Q', a \in \Sigma$, define $\delta'(q, a) \coloneqq \begin{cases}
\delta(q, a) & \text{if } q \neq \bot \mbox{ and } \delta(q, a) \notin F,\\                                                                                                                                                                                                
\bot & \text{otherwise}.                                                                                                                                                                                               \end{cases}$
 \item Define $F' \coloneqq \{\bot\}$, respectively $F' \coloneqq Q' \setminus \{\bot\}$
\end{enumerate}


The family of DWAs is closed under Boolean operations. For an $\omega$-language $L$, define a congruence $\trace{L} \subseteq \Sigma^* \times \Sigma^*$ where $u\trace{L}v \Leftrightarrow \forall \alpha \in \Sigma^\omega, u\alpha \in L$ iff $v\alpha \in L$. If $L$ is recognized by a DWA then this congruence has a finite index. We say that an $\omega$-language is \emph{weakly recognizable} if it is recognized by a DWA. The class $\BC(\ext(\REG))$ of finite Boolean combinations of languages in $\ext(\REG)$ is exactly the set of weakly recognizable languages.

\begin{remark}[The minimal DWA \cite{Loe01}]\label{rem:DWA}
  If for a weakly recognizable language $L$, $M$ is the index of the congruence defined above, then the language is recognized by a DWA $\aut{A} = (Q, \Sigma, q_0, \delta, F)$ with $|Q| = M$. Also, for each state $q \in Q$ there exists a word $u_q \in \Sigma^*$ such that for each $u \in \Sigma^*, \delta(q_0, u) = q$ iff $u \in [u_q]_{\trace{L}}$.
\hfill$\boxtimes$
\end{remark}

Turning to traces, let $I\subseteq \Sigma\times \Sigma$ denote an irreflexive\footnote{A relation $R$ is \emph{irreflexive} if for no $x$ we have $xRx$.}, symmetric \emph{independence} relation over an alphabet $\Sigma$, then $D \coloneqq \Sigma^2 \setminus I$ is the reflexive, symmetric \emph{dependence} relation over $\Sigma$. We refer to the pair $(\Sigma, I)$ as the \emph{dependence alphabet}. For any letter $a \in \Sigma$, we define $I_a \coloneqq \{b \in \Sigma \mid a I b\}$ and $D_a \coloneqq \{b \in \Sigma \mid a D b\}$. A \emph{trace} can be identified with a labeled, acyclic, directed \emph{dependence graph} $[V, E, \lambda]$ where $V$ is a set of countably many vertices, $\lambda\colon V\to \Sigma$ is a labeling function, and $E$ is a countable set of edges such that, firstly, for every $v_1, v_2 \in V\colon \lambda(v_1) D \lambda(v_2) \Leftrightarrow (v_1, v_2) \in E \vee (v_2, v_1) \in E$; secondly, every vertex has only finitely many predecessors. $\mStar$ and $\mOmega$ represent the sets of all finite and infinite traces whose dependence graphs satisfy the two conditions above. We denote finite traces with the letter $t$, and an infinite trace with $\theta$; the corresponding languages with $T$ and $\Theta$ respectively. For a trace $t = [V, E, \lambda],$ define $\alphabet(t) \coloneqq \{a \in \Sigma \mid \emptyset \neq \lambda^{-1}(a) \subseteq V\}$, and similarly for a trace $\theta$. For an infinite trace, define $\alphinf(\theta) \coloneqq \{a \in \Sigma \mid |\lambda^{-1}(a)| = \infty\}$.

For two traces $t_1, t_2,\ t_1 \sqsubseteq t_2$ (or $t_1 \sqsubset t_2$) denotes that $t_1$ is a (proper) prefix of $t_2$. We denote the prefix relation between words similarly. The least upper bound of two finite traces, whenever it exists, denoted $t_1 \sqcup t_2$ is the smallest trace $s$ such that $t_1 \sqsubseteq s$ and $t_2 \sqsubseteq s$. Whenever it exists, one can similarly refer to the least upper bound $\bigsqcup S$ of a finite or an infinite set $S$ of traces. The concatenation of two traces is denoted as $t_1 \odot t_2$. Note that for any $t, \theta$ the concatenation $t\odot \theta \in \mOmega$. However, $\theta \odot t \in \mOmega$ iff $\alphinf(\theta) I \alphabet(t)$.

The canonical morphism $\Gamma\colon \Sigma^* \to \mStar$ associates finite words with finite traces, and the inverse mapping $\Gamma^{-1}\colon \mStar \to 2^{\Sigma^*}$ associates finite traces with equivalence classes of words. The morphism $\Gamma$ can also be extended to a mapping $\Gamma\colon \Sigma^\omega \to \mOmega$. For a (finite or infinite) trace $t$, the set $\Gamma^{-1}(t)$ represents the \emph{linearizations} of $t$. Two words $u, v$ are equivalent, denoted $u \trace{I} v$, iff $\Gamma(u) = \Gamma(v)$. We note that for finite traces the relation $\trace{I}$ coincides with the reflexive, transitive closure of the relation $\set{(uabv,ubav)\mid u,v\in\Sigma^*\wedge aIb}$. For a word $w$, define the set $[w]_{\trace{I}} \coloneqq \Gamma^{-1}(\Gamma(w))$. Finally, we say that a word language $K$ is \emph{trace-closed} iff $K = [K]_{\trace{I}}$, where $[K]_{\trace{I}} \coloneqq \bigcup_{u \in K} [u]_{\trace{I}}$. 

\begin{definition}
\label{def:recognizable}
 A trace language $T \subseteq \mStar$ (resp. $\Theta \subseteq \mOmega$) is \emph{recognizable} or \emph{regular} iff $\Gamma^{-1}(T)$ (resp. $\Gamma^{-1}(\Theta)$) is a recognizable word language.
\end{definition}
\noindent With $\rec(\mStar)$ and $\rec(\mOmega)$ we denote the classes of recognizable languages of finite and infinite traces respectively. 

Asynchronous cellular automata have been introduced \cite{DM93,GP92} as acceptors of \omReg trace languages. However, a global view of their (local) transition relations yields a notion of automata that recognize trace-closed word languages. Throughout this paper, we take this global view of asynchronous automata. Formally, a \emph{deterministic asynchronous cellular automaton (DACA)} over $(\Sigma,I)$ is a 4-tuple $\aut{a}=(\prod_{a\in\Sigma} Q_a,(\delta_a)_{a\in\Sigma},q_0,F)$, where $q_0\in \prod_{a\in\Sigma} Q_a$,  $\mapping{\delta_a}{\prod_{b\in D_a} Q_b}{Q_a}$ and $F\subseteq \prod_{a\in\Sigma}Q_a$. Given a state $q\in\prod_{a\in\Sigma} Q_a$ and a letter $b\in\Sigma$, the unique $b$-sucessor $\delta(q,b)= q' = (q'_a)_{a \in \Sigma} \in \prod_{a\in\Sigma} Q_a$  is given by $q'_b=\delta_b((q_a)_{a\in D_b})$ and $q'_a=q_a$ for all $a\neq b$. That is, the only component that changes its state is the component corresponding to $b$. Given a word $u\in\Sigma^*$ the \emph{run} $\rho_u$ of $\aut{a}$ on $u$ is given as usual by $\rho_u(0)=q_0$ and $\rho_u(i+1)=\delta(\rho_u(i),u[i])$. This definition extends naturally to infinite runs $\rho_{\alpha}$ on infinite $\alpha\in\Sigma^\omega$. A \emph{deterministic asynchronous Muller automaton (DACMA)} is an asynchronous automaton $\aut{a}=(\prod_{a\in\Sigma} Q_a,(\delta_a)_{a\in\Sigma},q_0,\F)$ with $\F\subseteq \prod_{a\in\Sigma}\Pow(Q_a)$. We define $\occset_a(\rho)$ of (a finite or an infinite) run $\rho$ to be the set $\set{\rho(0)_a,\rho(1)_a,\ldots}\subseteq Q_a$. Likewise, $\infset_a(\rho)=\set{q\in Q_a\mid \exists^\infty n\colon \rho(n)_a=q}$. A DACMA \emph{accepts} $\alpha\in\Sigma^\omega$ if for some $F=(F_a)_{a\in\Sigma}\in\F$ we have $\infset_a(\rho_{\alpha})=F_a$. 

 A word automaton $\aut{A}=(Q,\Sigma,q_0,\delta)$ is called \emph{$I$-diamond} if for every $(a,b) \in I$ and every state $q \in Q$, $\delta(q, ab) = \delta(q, ba)$. Every $T \in \rec(\mStar)$ (resp. $\Theta \in \rec(\mOmega)$) is recognized by a DACA \cite{DR95} (resp. a DACMA \cite{DM93}). Via their global behaviors, they accept the corresponding trace-closed languages, and in particular, every regular trace-closed language (resp. trace-closed \omReg language) is recognized by an $I$-diamond DFA (resp. $I$-diamond Muller automaton). In fact for every trace-closed $K \in \REG$, the minimal DFA $\aut{A}_K$ accepting $K$ is $I$-diamond.


Finally, we want to recall some basic algebraic definitions. Given a language $T$ of finite traces, a semigroup  $S$, and a morphism $\mapping{\varphi}{\mStar}{S}$, $\varphi$ is said to \emph{recognize} $T$ if there exists $P\subseteq S$ with $T=\varphi^{-1}(P)$. By extension, $S$ is said to recognize $T$ if such a morphism exists. A \emph{linked  pair} of a semigroup is a tuple $(s,e)\in S^2$ with $s\cdot e = s$ and $e\cdot e=e$. We state a well known consequence of Ramsey's theorem: Let $A$ be a (possibly infinite) alphabet, $S$ be any finite semigroup and  $\mapping{f}{A^+}{S}$ any mapping. Given an infinite sequence $\alpha\in A^\omega$ and an arbitrary factorization $\alpha=(u_i)_i$ of $\alpha$ into words $u_i\in A^+$, there exists a linked pair $(s,e)$ and a strictly monotone sequence $(n_i)_i$ of natural numbers with the property that $f(u_0\cdots u_{n_0})=s$ and $f(u_{n_i} \cdots u_{n_{i+1}-1})=e$ for all $i\in\naturals$. Let $(u'_i)_i$ be given by $u'_0=u_0\cdots u_{n_0}$ and $u_i'=u_{n_i}\cdots u_{n_{i+1}-1}$ for $i\geq 1$. We say this \emph{superfactorization}  is \emph{associated} with $(s,e)$. We will often use Ramsey's theorem implicitly. Given a semigroup $S$, a morphism $\mapping{\varphi}{\mStar}{S}$ is said to \emph{saturate} $\Theta\subseteq \mOmega$ if for every linked pair $(s,e)$ of $S$ we have either $\varphi^{-1}(s)\odot (\varphi^{-1}(e))^\omega\cap\Theta=\emptyset$ or $\varphi^{-1}(s)\odot (\varphi^{-1}(e))^\omega\subseteq \Theta$. Let  $\Theta$ be a language of infinite traces, $S$ be a finite semigroup, and  $\mapping{\varphi}{\mStar}{S}$ a saturating morphism. Then $\varphi$ \emph{recognizes} $\Theta$, if for some set $P$ of linked pairs of $S$ we have $\Theta=\bigcup_{(s,e)\in P}\varphi^{-1}(s)\odot (\varphi^{-1}(e))^\omega$. Again, we say $S$ recognizes $\Theta$ if such a morphism exists. These notions of recognizability coincide with the corresponding notions from Def.~\ref{def:recognizable}.
\section{From Regular Trace Languages to \boldmath{$\omega$}-Regular Trace Languages}\label{sec:tcLanguages}
We wish to extend the well-studied relations between regular and \omReg languages to the field of finite and infinite traces. We first look at reachability and safety languages, their Boolean combinations, i.e. the weakly recognizable languages, and study how they can be obtained as a result of infinitary operations on regular trace languages. We will later see that the case of Büchi recognizability is not straight forward. Our definitions are consistent with those over word languages; that is, if the dependence relation over the alphabet is complete then these definitions coincide.

\subsection{Infinitary Extensions of Regular Trace Languages}\label{subsec:wLanguages}

In the classification hierarchy of \omReg languages, reachability and safety languages occupy the lowest levels. For trace languages we have the following. 

\begin{definition}\label{def:extT}
 Let $T \in \rec(\mStar)$. The \emph{infinitary extension} is the $\omega$-trace language given by $\ext(T) \coloneqq \bigcup_{t \in T} t\odot \mOmega$. 
\end{definition}

However, the definition of infinitary extensions of a trace-closed languages is not sound with respect to trace equivalence of $\omega$-words; i.e. if $T \in \rec(\mStar)$ and $K = \Gamma^{-1}(T)$, then, in general, $\ext(K) \neq \Gamma^{-1}(\ext(T))$. 

\begin{example}\label{ex:Ext}
 Let $\Sigma = \{a, b, c\}$, and $b I c$. Define $K \coloneqq [ab]_{\trace{I}}$. Clearly $K$ is trace-closed and, moreover, $acb \notin K$. Let $T = \Gamma(K)$. Clearly $abc^\omega, acbc^\omega, accbc^\omega,\dots$ are equivalent words since they induce the same infinite trace which belongs to $\ext(T)$. However, while $abc^\omega \in \ext(K)$, $ac^+bc^\omega \nsubseteq \ext(K)$. 
\hfill $\boxtimes$
\end{example}

\begin{definition}\label{def:iSuffExt}
 Let $K \subseteq \Sigma^*$ be trace-closed. Define the \emph{$I$-suffix extended} trace-closed language (or \emph{$I$-suffix extension}) of $K$ as $ K_I \coloneqq K \cup \bigcup_{a \in \Sigma} [Ka^{-1}aI_a^*]_{\trace{I}}$.
\end{definition}

Due to the closure of $\rec(\mStar)$ under concatenation and finite union \cite{DR95}, we know that $K_I$ is regular whenever $K$ is regular.


\begin{proposition}\label{prop:extIsuff}
 If $T \in \rec(\mStar)$, $K = \Gamma^{-1}(T)$, and $K_I$ is the $I$-suffix extension of $K$, then $\Gamma^{-1}(\ext(T)) = \ext(K_I)$.
\end{proposition}
\begin{proof}
From the definitions of $K_I$ and $T$, we trivially observe that for every $\alpha \in \ext(K_I)$ it holds that $\Gamma(\alpha) \in \ext(T)$. Therefore, $\ext(K_I) \subseteq \Gamma^{-1}(\ext(T))$.

To show $\Gamma^{-1}(\ext(T)) \subseteq  \ext(K_I)$, we show that: (1) for every infinite trace in $\ext(T)$, there exists a linearization in $\ext(K_I)$; (2) the language $\ext(K_I)$ is trace-closed. 

\noindent (1) Consider $\theta \in \ext(T)$. Hence there exist $t \in T$ and $\theta' \in \mOmega$ such that $\theta = t\odot\theta'$. From the definitions, it follows that for any $w \in \Gamma^{-1}(t)$ and $\beta \in \Gamma^{-1}(\theta'), w\cdot\beta \in \ext(K)$ and therefore in $\ext(K_I)$.

\noindent (2) Let $\alpha \in \ext(K_I)$, and $t \in T$ be a trace such that $t \sqsubset \Gamma(\alpha)$. Consider any $\beta \in \Sigma^\omega$ such that $\beta \trace{I} \alpha$. Trace equivalence implies that $t \sqsubset \Gamma(\beta)$. Moreover there exists a minimal natural number $i \in \naturals, t \sqsubseteq \Gamma(\beta[1,i])$. Observe that $\beta[i]$ is a maximal symbol appearing in $t$ because otherwise we can contradict the minimality of $i$ and find $i' < i$ such that $t \sqsubseteq \Gamma(\beta[1,i'])$. Now, let $s \in \mStar$ be the finite trace such that  $t\odot s = \Gamma(\beta[1,i])$. 

It must hold that either $s$ is the empty trace or $\beta[i]\times \alphabet(s) \subseteq I$, because otherwise $t\odot s \neq \Gamma(\beta[1,i])$. This implies $\beta[1,i] \in K_I$, and hence $\beta \in \ext(K_I)$.
\end{proof}

\begin{remark}
In general $K_I\neq (K_I)_I$. However,  iterated $I$-suffix extensions preserve the infinitary extension languages: $\ext(K) \subseteq \ext(K_I) = \ext((K_I)_I) \dots$ 
\hfill$\boxtimes$ 
\end{remark}

Proposition \ref{prop:extIsuff} provides us the basis for generating the class of weakly recognizable trace-closed languages corresponding to the recognizable subset of $\BC(\ext(\mStar))$. Henceforth, whenever we speak of the language $\Gamma^{-1}(\ext(T))$ we refer to $\ext(\Gamma^{-1}(T)_I)$. Similarly, for a trace-closed language $K$ we always mean $\ext(K_I)$ whenever we say $\ext(K)$.

\begin{theorem}\label{thm:DWA}
 A trace-closed language $L \subseteq \Sigma^\omega$ is recognized by an $I$-diamond DWA if and only if $L \in \BC(\ext(\K))$ for a finite set $\K \subseteq 2^{\Sigma^*}$ of trace-closed regular languages.
\end{theorem}
\begin{proof}
Given trace-closed regular languages $K \in \K$, we construct $I$-diamond DWA $\aut{A}_K$ accepting $\ext(K)$ as mentioned previously. Let $L \coloneqq \bigcup_{i} (\bigcap_{j} L_{i,j})$ be the language expressed in disjunctive normal form over $\ext(\K)$ (for each $i, j,\ L_{i,j}$ is either of the form $\ext(K)$ or $\overline{\ext(K)}$). We define the product DWA $\aut{A} \coloneqq (\prod_{K \in \K} Q_K, \Sigma, (q_0^K)_{K \in \K}, \delta, F)$ where:
\begin{itemize}
 \item $\delta((p^K)_{K\in\K}, a) = (q^K)_{K\in\K}$ if and only if $\delta_K(p^K, a) = q^K$ for all $K \in \K$
 \item The tuple $(q^K)_{K\in\K} \in F$ if and only if it satisfies some conjunct. That is, for some $i$ it holds that whenever $L_{i,j} = \ext(K)$ then $q^K = \bot_K$, and whenever $L_{i,j} = \overline{\ext(K)}$ then $q^K \neq \bot_K$ for all $K \in \K$. 
\end{itemize}
It is easily verified that $\aut{A}$ is an $I$-diamond DWA accepting $L$.

For the other direction, consider the minimal DWA $\aut{A} = (Q, \Sigma, q_0, \delta, F)$ that accepts $L$. Since trace equivalence $\trace{I}$ over finite words is a finer congruence than the language congruence $\trace{L}$ (i.e. $u \trace{I} v \Rightarrow u \trace{L} v$ for all $u, v \in \Sigma^*$), it follows that for any pair of finite trace equivalent words $u, v \in \Sigma^*, \delta(q_0, u) = \delta(q_0, v)$. Thus, $\aut{A}$ is $I$-diamond. 

For each SCC $S \subseteq Q$ of $\aut{A}$, let $K_S \in REG$ trace-closed be the language accepted by $\aut{A}_S \coloneqq (Q, \Sigma, q_0, \delta, S)$. Recall that each SCC of a DWA contains either only accepting states or rejecting states. Then, the language $L$ accepted by $\aut{A}$ is given by the following disjunction over all accepting SCC's $L \coloneqq  \bigcup_{S} L_S$, where $L_S \coloneqq \ext(K_S) \cap \bigcap_{S' \neq S} \overline{\ext(K_{S'})}$.
\end{proof}
\subsection{Infinitary Limits of Regular Trace Languages} \label{subsec:sLanguages}

We now consider the \emph{infinitary limit} operator. In the case of word languages, this operator extends regular languages to the family \omReg languages that are DBA recognizable. In particular, we seek an effective characterization of languages $T \in \rec(\mStar)$, such that $\Gamma^{-1}(\lim(T))$ is recognized by an $I$-diamond DBA.

\begin{definition}\label{def:limT}
 Let $T \in \rec(\mStar)$, the \emph{infinitary limit} $\lim(T)$ is the $\omega$-trace language containing all  $\theta \in \mOmega$ such that there exists a sequence $(t_i)_{i \in \naturals}, t_i \in T$ satisfying $t_i \sqsubset t_{i+1}$ and $ \bigsqcup_{i \in \naturals} t_i = \theta$.
\end{definition}

\begin{remark}
For $T \in \rec(\mStar)$, it holds that $\lim(T) \in \rec(\mOmega)$. In fact, if for a finite semigroup $S$, a morphism $\varphi \colon \mStar \to S$ recognizes $T$, then $\lim(T)$ can be described in terms of a set $P_T$ of linked pairs of $S$, i.e. $\lim(T) = \bigcup_{(s,e) \in P_T}\varphi^{-1}(s)\odot (\varphi^{-1}(e))^\omega$.
\hfill$\boxtimes$
\end{remark}

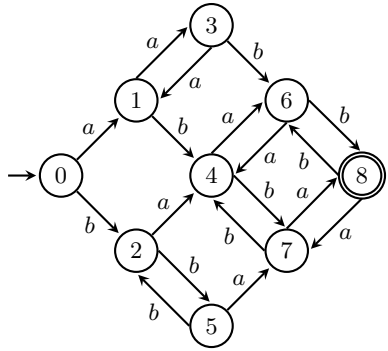
\begin{figure}[ht]
\centering
 	\begin{tikzpicture}[node distance=10mm, shorten >=1pt, thick, >=stealth, every state/.style={minimum size=3mm}, initial text=]
	  \node[state, initial] (q) at (0,0)  {\small{$0$}};
	  \node[state] (p1) at (1,1) {\small{$1$}};
	  \node[state] (p2) at (1,-1) {\small{$2$}};
	  \node[state] (p3) at (2,2) {\small{$3$}};
	  \node[state] (p4) at (2,0) {\small{$4$}};
	  \node[state] (p5) at (2,-2) {\small{$5$}};
	  \node[state] (p6) at (3,1) {\small{$6$}};
	  \node[state] (p7) at (3,-1) {\small{$7$}};
	  \node[state, accepting] (p8) at (4,0) {\small{$8$}};

	  \path[->]  
			(q) edge node [auto, inner sep=1pt]  {\small{$a$}} (p1)
			(q) edge node [auto, swap, inner sep=1pt]  {\small{$b$}} (p2)
			
			(p1.north) edge node[auto, inner sep=1pt] {\small{$a$}} (p3.west)
			(p1) edge node[auto, inner sep=1pt] {\small{$b$}} (p4)

			(p2.east) edge node[auto, inner sep=1pt] {\small{$b$}} (p5.north)
			(p2) edge node[auto, inner sep=1pt] {\small{$a$}} (p4)

			(p3.south) edge node[auto, inner sep=1pt] {\small{$a$}} (p1.east)
			(p3) edge node[auto, inner sep=1pt] {\small{$b$}} (p6)

			(p4.east) edge node[auto, inner sep=1pt] {\small{$b$}} (p7.north)
			(p4.north) edge node[auto, inner sep=1pt] {\small{$a$}} (p6.west)

			(p5.west) edge node[auto, inner sep=1pt] {\small{$b$}} (p2.south)
			(p5) edge node[auto, inner sep=1pt] {\small{$a$}} (p7)

			(p6.south) edge node[auto, inner sep=1pt] {\small{$a$}} (p4.east)
			(p6.east) edge node[auto, inner sep=1pt] {\small{$b$}} (p8.north)

			(p7.west) edge node[auto, inner sep=1pt] {\small{$b$}} (p4.south)
			(p7.north) edge node[auto, inner sep=1pt] {\small{$a$}} (p8.west)

			(p8.south) edge node[auto, inner sep=1pt] {\small{$a$}} (p7.east)
			(p8.west) edge node[auto, inner sep=1pt] {\small{$b$}} (p6.south);
	\end{tikzpicture}   
%
%
%
%
\caption{The minimal DFA recognizing language $K$ of Example \ref{ex:noLim}.}
\label{fig:noLim}
\end{figure}

\begin{example}\label{ex:noLim}
 Let $\Sigma = \{a, b\}$, and $I = \{(a, b), (b, a)\}$. Define $K \coloneqq [(aa)^+(bb)^+]_{\trace{I}}$ as the trace-closed language with even number of occurrences of $a$'s and $b$'s. The minimal DFA accepting this language is shown in Figure \ref{fig:noLim}. If $T = \Gamma(K)$, then $\lim(T)$ is defined as 
  \[\Theta = \Bigg{\{}\theta \in \mOmega 
      \left| \begin{array}{l}
       |\theta|_a \mbox{ even, } |\theta|_b = \infty \mbox{, or}\\
       |\theta|_a = \infty, |\theta|_b \mbox{ even, or}\\ 
       |\theta|_a = |\theta|_b = \infty
      \end{array} \right\}\]

The trace-closed language $L = \Gamma^{-1}(\Theta)$ consists of all infinite words $\alpha \in \Sigma^\omega$ that satisfy the same conditions as $\theta \in \Theta$ above.
\hfill $\boxtimes$
\end{example}

It is easy to verify that the DFA of Figure \ref{fig:noLim} does not accept $L$  when equipped with a Büchi acceptance condition. For instance, the automaton can loop forever in states $4$, $6$, and $7$, thereby witnessing infinitely many $a$'s and $b$'s, without ever visiting state $8$.

\begin{proposition}\label{prop:noLim}
 There does not exist any $I$-diamond DBA recognizing $L \subseteq \Sigma^\omega$ as described in Example \ref{ex:noLim}.
\end{proposition}
\noindent A proof of this proposition can be found in the appendix.

\begin{corollary}
 There exists a family $\K$ of trace-closed regular languages of finite words, namely $\K \coloneqq \{[(a^m)^+ (b^n)^+]_{\trace{I}} \mid m, n \geq 2\}$ over $\Sigma = \{a, b\}$, such that given $T = \Gamma(K)$ for any $K \in \K$, there exists no $I$-diamond DBA recognizing $\Gamma^{-1}(\lim(T))$.
\end{corollary}

\begin{definition}
 A trace-closed language $K \subseteq \Sigma^*$ is \emph{$I$-limit-stable} (or simply \emph{limit-stable}) if $\lim(K)$ is also trace-closed. By extension, $T \subseteq \mStar$ is \emph{limit-stable} if $\Gamma^{-1}(T)$ is.
\end{definition}

Toward characterizing limit-stable languages, we introduce some definitions. Let $T\subseteq\mStar$ be a language of traces and let $t\pprefix t'$ be two traces. The \emph{prefix graph} of the pair $(t,t')$ is the directed, acyclic graph $G_{t,t'}=(V,E)$ with $V=\set{x\in\mStar\mid t\prefix x\prefix t'}$ and $(x,y)\in E$ if $y=x\odot a$ for some $a\in\Sigma$. A \emph{cut} of $G_{t,t'}$ is a set $C\subseteq V\setminus\set{t,t'}$ such that each path from $t$ to $t'$ in $G_{t,t'}$ visits at least one vertex from $C$. Note that if $t'=t\odot a$ for some $a\in\Sigma$, then $G_{t,t'}$ does not admit a cut. A pair $(t,t')$ is \emph{$T$-separable} if $G_{t,t'}$ admits a cut $C\subseteq T$.

Let $\theta\in\lim(T)$. Define an infinite transition-graph $G=G_{\theta}=(V,\Delta)$ with $V=\set{t\in\mStar\mid t\prefix \theta}$ and $(t,a,t')\in \Delta$ if $t'=t\odot a$ for some $a\in\Sigma$. Then there is a one to one correspondence between the paths starting from $\epsilon$ through $G$ and the linearizations of $\theta$. More precisely, for any finite word $u\in\Sigma^*$, there exists a run $\rho_u$ from $\epsilon$ on $u$ in $G_{\theta}$ iff $u$ is the linearization of some prefix $t$ of $\theta$. An infinite word $\alpha$ is a linearization of $\theta$ iff $\alpha[1,n]$ is a linearization of some prefix $t_n$ of $\theta$ for all $n\in\naturals$. Hence, an $\omega$-word $\alpha$ is a linearization of $\theta$ iff it induces a run $\rho_{\alpha}$ in $G_{\theta}$. 

Let $S$ be a finite semigroup, let $P\subseteq S$, and let $(s,e)$ be a linked pair of $S$. Let $\varphi$ be a morphism from $\mStar$ onto $S$. The pair $(s,e)$ has the \emph{$P$-cut property} if
\begin{itemize}
\item either for every factorization $\varphi(a_1)\cdots \varphi(a_k)=e$ with $a_i\in\Sigma$, we have $e\varphi(a_1\odot\cdots\odot a_j)\in s^{-1}P$ for some $j \in [1,k]$;
\item or for every factorization $\varphi(a_1)\cdots \varphi(a_k)=e$ with $a_i\in\Sigma$, we have $e\varphi(a_1\odot\cdots\odot a_j)\notin s^{-1}P$ for all $j \in [1,k]$.
\end{itemize}

\begin{lemma}
\label{lem:exists_satur_and_rec_semigrp}
Let $T \in \rec(\mStar)$. Then there exists a finite semigroup $S$ and a saturating morphism $\mapping{\alpha}{\mStar}{S}$ which recognizes both $\lim(T)$ and $T$. 
\end{lemma}
Such a morphism is said to \emph{simultaneously recognize} $T$ and $\lim(T)$. Given an automaton, we write $p \xrightarrow{u} q$ if some $u \in \Sigma^*$ leads from $p$ to $q$, and $p \xRightarrow{u} q$ if a final state is also visited.

\begin{definition}
Given $(\Sigma,I)$, let $\aut{A}=(Q,\Sigma,q_0,\delta,F)$  be an $I$-diamond automaton. $\aut{A}$ is \emph{$F,I$-cycle closed}, if for all $u\trace{I} v$ and all $q$ we have $q\xRightarrow{u}q$ iff $q\xRightarrow{v}q$.
\end{definition}

We can now give an effective characterization of limit-stable languages. Due to space constraints, we only present a part of the following proof here. Lem. \ref{lem:exists_satur_and_rec_semigrp} ensures that (\ref{thm:lim_trace_char:5}) is not trivially satisfied.

\begin{theorem}\label{thm:lim_trace_char}
Let $T \in \rec(\mStar)$ and let $K=\Gamma^{-1}(T)$. The following are equivalent:
\begin{enumerate}[(a)]
\item $K$, and therefore $T$, is limit-stable.\label{thm:lim_trace_char:1}
\item For all sequences $(t_i)=t_0\pprefix t_1\pprefix t_2\cdots \subseteq T$ and all sequences $(u_i)_i$ with $u_i\in\Gamma^{-1}(t_i)$, there exists a subsequence $(u_{j_i})_i$ and a sequence $(v_{j_i})_i$ of proper prefixes $v_{j_i}\pprefix u_{j_i}$ with $|v_{j_i}|<|v_{j_{i+1}}|$ and $v_{j_i}\in K$ for all $i\in\naturals$.\label{thm:lim_trace_char:2}
\item For any $\theta \in \lim(T)$ there exists a strictly monotone $(n_i)_i$ such that any infinite path $\rho$ in $G_\theta$ visits $T$ in each segment $\rho(n_i,n_{i+1}-1)$.\label{thm:lim_trace_char:3}
\item Let $(t_i)_{i}$ be a sequence of traces in $T$. Then there exists a subsequence $(t_{m_i})_i$, such that $(t_{m_i},t_{m_{i+1}})$ is $T$-separable for all $i$.\label{thm:lim_trace_char:4}
\item If $T$ and $\lim(T)$ are simultaneously recognized by a morphism $\mapping{\varphi}{\mStar}{S}$ for some finite semigroup $S$, then every linked pair $(s,e)$ has the $\varphi(T)$-cut property. \label{thm:lim_trace_char:5}  
\item Any DFA $\aut{A}$ recognizing $K$ is $F,I$-cycle closed.\label{thm:lim_trace_char:6}  
\end{enumerate}
\end{theorem}
\begin{proof}
(\ref{thm:lim_trace_char:1})$\implies$(\ref{thm:lim_trace_char:2}): If (\ref{thm:lim_trace_char:2}) is false, then we may choose a sequence $(t_i)_i$ of traces in $T$ with the property that for some sequence $(u_i)_i$ of linearizations of $(t_i)_i$, every subsequence $(u_{n_i})_i$, and every sequence $(v_{n_i})_i$ of proper prefixes $v_{n_i}\pprefix u_{n_i}$, $v_{n_i}\in K$, we have $\sup_i|v_{n_i}|<\infty$. Since $|\Sigma|<\infty$ we have that $\Sigma^{\infty}$ is a compact space. Hence $(u_i)_i$ has a converging subsequence $(u_{m_i})_i$. Because every subsequence of $(u_i)_i$ has the properties given in the previous sentence, so does $(u_{m_i})_i$. Let $\alpha=\lim_{i\to \infty} u_{m_i}$. Then $\alpha\trace{I} \beta$ for some $\beta= x\cdot y_1\cdot y_2\cdots$ with $x \cdot y_1\cdots y_i \in\Gamma^{-1}(t_{m_i})$. Hence, $\beta\in \lim(L)$. But, by construction, $\alpha\notin \lim(K)$ because for some $n\in\naturals$ no prefix of length $>n$ of $\alpha$ is in $K$.

(\ref{thm:lim_trace_char:2})$\implies$(\ref{thm:lim_trace_char:1}): Let $\theta=\bigsqcup_i t_i$ for traces $t_i\in T$. We may assume that $t_i\pprefix t\pprefix t_{i+1}$ implies $t\notin T$. Let $\alpha\in\Gamma^{-1}(\theta)$. Then we pick prefixes $(w_i)_i$ of $\alpha$, such that $w_i$ is of minimal length with $t_i\prefix \Gamma(w_i)$. Consider the subsequence $(t_{2i})_i$ of $(t_i)_i$. Each $w_{2i+1}$ is a prefix of some linearization of $t_{2(i+1)}$, say $u_{2(i+1)}$. We apply (\ref{thm:lim_trace_char:2}) to the sequence $(t_{2i})_i$ and get a sequence $(v_{2i})_i$ of proper prefixes of the $u_{2i}$, such that $\sup_i|v_{2i}|=\infty$ and $v_{2i}\in K$. We now have to show that $v_{2i}$ is already a prefix of $w_{2i-1}$. Suppose not, i.e. $w_{2i-1}\pprefix v_{2i}\pprefix u_{2i}$. Then this would give a trace $t \in T$ with $t_{2i-1}\pprefix t\pprefix t_{2i}$. 

(\ref{thm:lim_trace_char:1})$\implies$(\ref{thm:lim_trace_char:6}): Suppose $\aut{A}$ is not $I$-cycle closed. Then there exists $q\in Q$ and $u\trace{I} v$ with $q\xRightarrow{u}q$ but not $q\xRightarrow{v}q$. Since $\aut{A}$ is $I$-diamond, this means that the run $q\xrightarrow{v}q$ exists, but does not visit a final state. Now pick $x\in\Sigma^*$ with $q_0\xrightarrow{x}q$. Then $\alpha=x\cdot u^\omega\in \lim(K)$ and $\beta=x\cdot v^\omega\notin \lim(L)$. But clearly $\alpha\trace{I}\beta$ implies that $\lim(K)$ is not trace-closed.

(\ref{thm:lim_trace_char:6})$\implies$(\ref{thm:lim_trace_char:1}): Let $\alpha\trace{I} \beta$ and let $\alpha\in\lim(K)$. Take $\aut{A}=\aut{A}_K$ and consider extended transition profiles $\tau_w\subseteq Q\times\set{0,1}\times Q$ for $w\in\Sigma^*$ defined by $(p,1,q)\in\tau_w$ iff $p\xRightarrow{w} q$ and $(p,0,q)\in\tau_w$ iff $p\xrightarrow{w} q$ but not $p\xRightarrow{w} q$. Then we can factorize $\alpha=uv_0v_1v_2\cdots$ for finite words $u,v_0,v_1,\ldots$ with $\tau_u\cdot\tau_{v_i}=\tau_u$ and $\tau_{v_i}\cdot \tau_{v_i}=\tau_{v_i}$. Likewise we can factorize $\beta=u'v_0'v_1'\cdots$. 

Next, we observe that we find $r\in\naturals$ with $\Gamma(u'v_0')\prefix \Gamma(uv_0\cdots v_r)$. This gives $x\in\Sigma^*$ with $u'v_0'\cdot x\trace{I} uv_0\cdots v_r$. Conversely, there exists $m\in\naturals$ with $\Gamma(uv_0\cdots v_{r+1})\prefix \Gamma(u'v_0'\cdots v_m')$ and therefore $y\in\Sigma^*$ with $u'v_0'\cdots v_m' \trace{I} uv_0\cdots v_rv_{r+1}y \trace{I} u'v_0'xv_{r+1}y$, which implies $xv_{r+1}y\trace{I} v_1'\cdots v_m'$. 

Notice that if $q_0\xrightarrow{u} q$ and  $q_0\xrightarrow{u'} q'$, then (by trace equivalence and the fact that $\aut{A}$ is $I$-diamond) we have $q'\xrightarrow{x} q$. Likewise we have $q\xrightarrow{y} q'$ and $q'\xrightarrow{xv_{r+1}y}q'$. Now we can apply (\ref{thm:lim_trace_char:6}) to see that $q'\xRightarrow{xv_{r+1}y}q'$ iff $q'\xRightarrow{v_1'\cdots v_m'}q'$. However, since $\alpha\in\lim(K)$, since $\tau_{v_{r+1}}=\tau_{v_i}$ for all $i$, and since $q\xRightarrow{v_{r+1}}q$,  we have  $q'\xRightarrow{xv_{r+1}y}q'$.  Hence, $q'\xRightarrow{v_1'\cdots v_m'}q'$. Since furthermore $\tau_{v_1'\cdots v_m'}=\tau_{v_i'}$, we have for all $i, q'\xRightarrow[F]{v_i'}q'$ whence $\beta\in\lim(K)$.
\end{proof}

\begin{corollary}\label{cor:I-DBA-decidable}
Let $K=\Gamma^{-1}(T)$ for some $T \in \rec(\mStar)$. Given $\aut{A}_K$, it is decidable in time $\bigO(|Q|^2\cdot |\Sigma|(|\Sigma|+\log|Q|))$ whether or not $K$ is limit-stable.
\end{corollary}

Let $L\subseteq \Sigma^\omega$ be recognizable, trace-closed. Pick a DACMA (c.f. Sec. \ref{prelim}) $\aut{a}$ recognizing $L$. Recall that the global transition behavior of $\aut{a}$ gives an $I$-diamond DFA, which we denote by $\aut{A}=(\prod_{a\in\Sigma} Q_a,\Sigma,q_0,\delta)$. Given $q\in Q_a$ we define the DBA $\aut{A}_q=(\prod_{a\in\Sigma} Q_a,\Sigma,q_0,\delta,F_q)$, where $F_q=\set{q}\times \prod_{b\neq a}Q_b$. Note that $\aut{A}_q$ is $F_q,I$-cycle closed, because for any $q'\in \prod_{a\in\Sigma} Q_a$ and all $u\trace{I} v$ with $q'\xrightarrow{u}q'$ and $q'\xrightarrow{v}q'$ we have\footnote{This can be proven by an induction on the number of swapping operations needed to obtain $v$ from $u$.} $\occset_a(q'\xrightarrow{u}q')=\occset_a(q'\xrightarrow{v}q')$.  Now: 
\begin{equation*}
L=\bigcup_{(F_a)_{a\in\Sigma}\in\F}\bigcap_{a\in\Sigma}\bigcap_{q\in F_a} L(\aut{A}_q)\cap \bigcap_{q\notin F_a} \overline{L(\aut{A}_q)}
\end{equation*}

In \cite{DM93}, it was shown using algebraic arguments that every \omReg trace language can be expressed as a finite Boolean combination of ``restricted'' $\lim$-languages. This result also extends to the corresponding trace-closed linearization languages. Our characterization of limits of limit-stable languages allows for a first automata-theoretic equivalence result.

\begin{theorem}
Let $L$ be a trace-closed $\omega$-language. $L$ is \omReg iff $L$ is a finite Boolean combination of $I$-diamond DBA recognizable trace-closed languages.
\end{theorem}
\section{Conclusion}\label{sec:concl}
The main contribution of this paper is a new setup for a classification theory of languages of infinite traces (motivated by the first two levels of the Borel hierarchy). For any $T\in \rec(\mStar)$ we investigated the relationship between its infinitary extension $\ext(T)$ and the infinitary extension $\ext(K)$, where $K=\Gamma^{-1}(T)$. We showed that any such $K$ can be modified to $K_I$ such that $\ext(K_I)$ is also trace-closed and thus corresponds to the linearizations of $\ext(T)$. Building on this correspondence, we characterized the class of $I$-diamond DWA recognizable trace-closed languages in terms of Boolean combinations of trace-closed extensions of languages from $\REG$. In a similar vein, we characterized the class of languages $T\in\rec(\mStar)$ for which the linearization language of $\lim(T)$ is recognizable by an $I$-diamond DBA obtained from the minimal DFA for $\Gamma^{-1}(T)$, called limit-stable languages. Moreover, we showed that this class of languages is a decidable, proper subclass of finite recognizable trace languages. We proved how every recognizable language of infinite traces is a Boolean combination of languages $\lim(T)$ for limit-stable languages $T$.

\vspace{25pt}
\noindent \textbf{Acknowledgement}\hspace{5pt} We would like to thank Christof L\"{o}ding and Wolfgang Thomas for encouragement and numerous fruitful discussions.

\bibliographystyle{plain}
\bibliography{main}

\newpage

\appendix
\section{Proofs}

\subsection{Proof of Proposition \ref{prop:noLim}}

\begin{proposition*}
 There does not exist any $I$-diamond DBA recognizing $L \subseteq \Sigma^\omega$ as described in Example \ref{ex:noLim}.
\end{proposition*}
\begin{proof}
  Firstly, verify that $L$ is an \omReg trace-closed language. The transition graph of Figure \ref{fig:noLim} can be equipped with Muller accepting conditions to recognize $L$, namely $\F \coloneqq \{\{6,8\}, \{7,8\}, \{4,6,7\}, \{4,6,8\}, \{4,7,8\}, \{6,7,8\}, \{4,6,7,8\}\}$. Also note that since the Muller sets are closed under supersets, $L$ is in fact recognized by some DBA.
 
 Now, let us assume that $L$ is also recognized by some $I$-diamond DBA $\aut{A}_L$ with $n$ states. Let $q_0$ be the initial state and $\delta$ be the transition function of this automaton. We consider the word $a^{2n+1}b^{2n+1}$. Let $p_1 = \delta(q_0, a^{2n+1})$, $p_2 = \delta(q_0, b^{2n+1})$, and $p_3 = \delta(q_0, a^{2n+1}b^{2n+1}) = \delta(q_0, b^{2n+1}a^{2n+1})$. Let $k_2$ be the smallest non-zero number such that $\delta(p_1, a^{k_2}) = p_1$. Then we can factorize $a^{2n+1}$ into $a^{k_1}a^{k_2}, k_1 + k_2 = 2n+1$. Now, for $b^{2n+1}$, let $\ell_2$ be the smallest non-zero number that yields the corresponding factorization $b^{\ell_1}b^{\ell_2}$ at state $p_3$. This is shown in Figure \ref{fig:noDetBuechi}, which shows the transition subgraph that must necessarily occur in the automaton. Along state $p_2$, we obtain another pair of factorizations with $k'_1 + k'_2 = \ell'_1 + \ell'_2 = 2n+1$. Moreover, our assumptions ensure that $k'_2 > 0$ and $\ell_2 > 0$. Now consider the following possibilities.

\begin{figure}[h]
 \centering
 	\begin{tikzpicture}[node distance=10mm, shorten >=1pt, thick, >=stealth, every state/.style={minimum size=3mm}, initial text=, transform shape]
	  \node[state, initial] (q) at (0,0)  {$q_0$};
	  \node[state] (p1) at (1.5,1.5) {$p_1$};
	  \node[state] (p2) at (1.5,-1.5) {$p_2$};
	  \node[state] (p4) at (3,0) {$p_3$};

	  \path[->]  
			(q) edge node [auto]  {$a^{k_1}$} (p1)
			  edge node [auto, swap]  {$b^{\ell_1'}$} (p2)
			
			(p1) edge node[auto] {$b^{\ell_1}$} (p4)
			 edge[loop] node[auto] {$a^{k_2}$} ()

			(p2) 	edge node[auto,swap] {$a^{k_1'}$} (p4)
			  	edge[in=225, out=315, loop,swap] node[auto] {$b^{\ell_2'}$} ()

			(p4) edge[in=90, out=5, loop] node[auto] {$b^{\ell_2}$} ()
			(p4) edge[in=270, out=355, loop] node[auto,swap] {$a^{k_2'}$} ();
	\end{tikzpicture}   
\caption{Behavior of any ${I}$-diamond DBA ${\aut{A}_L}$ over ${a^{2n+1}b^{2n+1}}$.}
\label{fig:noDetBuechi}
\end{figure}
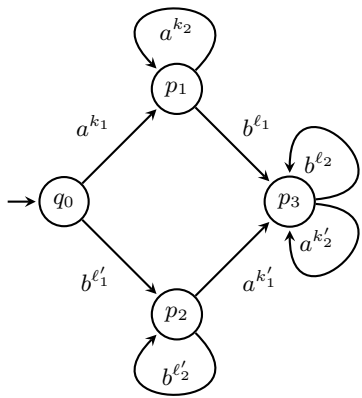

If $k_1$ is even, then the word $a^{k_1}b^{\ell_1}b^\omega \in L$. Therefore, the $\ell_2$-loop beginning at state $p_3$ must contain at least one Büchi accepting state. But then, the word $a^{k_1}a^{k_2}b^{\ell_1}b^\omega = a^{2n+1}b^\omega$ is also accepted, which is a contradiction. 

If $k_1$ is odd, then the $\ell_2$-loop beginning at state $p_3$ cannot contain any Büchi accepting states otherwise $a^{k_1}b^\omega$ will be accepted by the automaton. Now, $\ell_1$ can be either even or odd. In the former case, it must hold that the $k'_2$-loop beginning at state $p_3$ must contain a Büchi accepting state since $a^{k_1}b^{\ell_1}a^\omega \in L$. But then $a^{k_1}b^{\ell_1}b^{\ell_2}a^\omega$ will be accepted, leading to a contradiction. In the other case ($\ell_1$ odd), the $k'_2$-loop cannot contain any Büchi accepting states since $a^{k_1}b^{\ell_1}a^\omega \notin L$. But then, since the $\ell_2$-loop also does not have any accepting states, the word $a^{k_1}b^{\ell_1}(a^{k'_2}b^{\ell_2})^\omega \in L$ will also be rejected.    
\end{proof}

\subsection{Proof of Lemma \ref{lem:exists_satur_and_rec_semigrp}} 

\begin{lemma*}
Let $T\subseteq \mStar$ be a recognizable trace-language. Then there exists a finite semigroup $S$ and a morphism $\mapping{\alpha}{\mStar}{S}$ which saturates $\lim(T)$ and recognizes $T$. 
\end{lemma*}
\begin{proof}
There exists a finite semigroup $S'$ and a morphism $\varphi'$ which saturates $\lim(T)$. Furthermore, there exists a finite semigroup $S''$ and a morphism $\varphi''$ which recognizes $T$, say $T=\varphi''^{-1}(P)$. Let $S:=S'\times S''$ and $\varphi:=\varphi'\times\varphi''$. Then $\alpha$ recognizes $T$ and saturates $\lim(T)$. It remains to show that there stil exists a set of linked pairs $(s,e)$ of $S$ recognizing $\lim(T)$. To see this, pick any linked pair of $S'$, say $(s',e')$. Then any trace $\theta$ associated with this pair admits a factorization $xy_1y_2\cdots$ with $\varphi(x)=s$ and $\varphi(y_i)=e$. Now this factorization admits a superfactorization which is associated with a linked pair of $S''$. The claim now follows.
\end{proof}

\subsection{Full Proof of Theorem \ref{thm:lim_trace_char}}

\begin{theorem*}
Let $T\in\rec(\mStar)$ and let $K=\Gamma^{-1}(T)$. The following are equivalent:
\begin{enumerate}[(a)]
\item $K$, and therefore $T$, is limit-stable.\label{thm:lim_trace:1}
\item For all sequences $(t_i)=t_0\pprefix t_1\pprefix t_2\cdots \subseteq T$ and all sequences $(u_i)_i$ with $u_i\in\Gamma^{-1}(t_i)$, there exists a subsequence $(u_{j_i})_i$ and a sequence $(v_{j_i})_i$ of proper prefixes $v_{j_i}\pprefix u_{j_i}$ with $|v_{j_i}|<|v_{j_{i+1}}|$ and $v_{j_i}\in K$ for all $i\in\naturals$.\label{thm:lim_trace:2}
\item For any $\theta \in \lim(T)$ there exists a strictly monotone $(n_i)_i$ such that any inifnite path $\rho$ in $G_\theta$ visits $T$ in each segement $\rho(n_i,n_{i+1}-1)$.\label{thm:lim_trace:3}
\item Let $(t_i)_{i}$ be a sequence of traces in $T$. Then there exists a subsequence $(t_{m_i})_i$, such that $(t_{m_i},t_{m_{i+1}})$ is $T$-separable for every $i$.\label{thm:lim_trace:4}
\item If $T$ and $\lim(T)$ are simultaneoulsy recognized by a morphism $\mapping{\varphi}{\mStar}{S}$ for some finite semigroup $S$, then every linked pair $(s,e)$ has the $\varphi(T)$-cut property. \label{thm:lim_trace:5}  
\item Any DFA $\aut{A}$ recognizing $K$ is $F,I$-cycle closed.\label{thm:lim_trace:6}  
\end{enumerate}
\end{theorem*}
\begin{proof}
 
 (\ref{thm:lim_trace:1})$\implies$(\ref{thm:lim_trace:3}): Let $\theta\in\lim(T)$. If for every $n\in\naturals$ there exists a run $\rho_n$ through $G_{\theta}$ that visits a trace $t\in T$ only after $n$ positions, then there exists a run through $G_{\theta}$ which never visits a trace in $T$. This is because $(\rho_n)_n$ admits a converging subsequence (the space is compact) and because the set $[G_{\theta}]$ of all paths is closed and so this limit must itself be a path through $G_{\theta}$. This contradicts (\ref{thm:lim_trace:1}). Hence there exists $n_0$, such that every path through $G_{\theta}$ visits $T$ after at most $n_0$ steps. We now consider all finite segements of length $n_0$ and extend them. Let $U$ be the set of all those segements. Let $u\in U$. By a similar argument as before, there exists a number $n_u$, such that every extension $v=ux$ of length $n_1$ has visited $T$ at least once after $u$. Since there are finitely many segements in $U$, we can take the maxmimum $n_1=\max_{u\in U} n_u$. In this way we construct $(n_i)_i$. 

 (\ref{thm:lim_trace:3})$\implies$(\ref{thm:lim_trace:4}): Given $(t_i)_i\subseteq T$ we let $\theta=\bigsqcup_i t_i$ and pick $(n_i)_i$ as in (\ref{thm:lim_trace:3}). Now we pick $m_0$ arbitrary. Then, given $m_i$, we pick $m_{i+1}$, such that $|t_{m_{i+1}}|>\min \set{n_{j+1}\mid |t_{m_i}| < n_j}$. Now consider $(t_{m_i},t_{m_{i+1}})$. Because there exists  $n_j$ with $|t_{m_i}|<n_j < n_{j+1} < |t_{m_{i+1}}|$ we have that every path from $t_{m_i}$ to $t_{m_{i+1}}$ visits $T$ at least once. Hence $(t_{m_i},t_{m_{i+1}})$ is $T$-separable.

(\ref{thm:lim_trace:4})$\implies$(\ref{thm:lim_trace:5}): Let $\varphi$ and $S$ be as in the statement. Let $(s,e)$ be a linked pair. If $\lim(T)\cap \varphi^{-1}(s)\odot (\varphi^{-1}(e))^\omega= \emptyset$, then for every factorization $\varphi(a_1\odot\cdots \odot a_k)=e$ with $a_i\in\Sigma$ and every $i$ we have $e\varphi(a_1\odot\cdots\odot a_i)\notin s^{-1}P$. Indeed, if for some $\varphi(a_1\odot\cdots\odot a_k)=e$ we have $e\varphi(a_1\odot\cdots\odot a_i)\in s^{-1}P$, then the trace $x(a_1\odot\cdots\odot a_k)^n (a_1\odot\cdots\odot a_i)\in T$ for every $x\in\varphi^{-1}(s)$ and $n\in\naturals$. This contradicts the premise. 

Now if $\lim(T)\supseteq \varphi^{-1}(s)\odot(\varphi^{-1}(e))^\omega$ we pick an arbitrary factorization $\varphi(a_1\odot \cdots \odot a_k)=e$ and consider the sequence $(t_i)_{i}$ of traces given by $t_0=x\in\varphi^{-1}(s)$ and $t_{i+1}=t_i\odot a_1\odot\cdots \odot a_k$. Then by (\ref{thm:lim_trace:4}) there exists a subsequence $(t_{a_i})_{i}$, such that the pair $(t_{a_i},t_{a_{i+1}})$ is stable. Since $\varphi(t_{a_i})=s=se$ for all $i$, this implies that $x (a_1\odot\cdots \odot a_k)(a_1\odot\cdots \odot a_k)^{r}\odot a_1\odot\cdots \odot a_j\in T$ for some $1\leq j\leq k$ and $0\leq r < a_{i+1}-a_i$. Hence, $ \varphi((a_1\odot\cdots \odot a_k)^{r+1}\odot a_1\odot\cdots \odot a_j)=e\varphi(a_1\odot\cdots \odot a_j)\in s^{-1}P$.

(\ref{thm:lim_trace:5})$\implies$(\ref{thm:lim_trace:1}):  By lemma \ref{lem:exists_satur_and_rec_semigrp}, we may pick a finite semigroup $S$, a subset $P$ of $S$ and a morphism $\varphi$ from $\mStar$ onto $S$ which recognizes $T$ and saturates $\lim(T)$. By (\ref{thm:lim_trace:5}) every linked pair has the $P$-cut property. Let $\alpha\in\Gamma^{-1}(\theta)$  for some $\theta\in\lim(T)$. We may factorize $\theta=\alpha(0)\odot \alpha(1)\odot \cdots$. Let $(s,e)$ be a linked pair associated with a superfactorzation of this factorization and denote the corresponding factorization of $\alpha$ by $\alpha=u v_0 v_1 v_2\cdots$. Let $v_i=v_{i1}\cdots v_{ik_i}$ with $v_{ij}\in\Sigma$. Then, because $\varphi(\Gamma(v_i))=e$ and because $(s,e)$ has the $P$-cut property, the factorization $e=\varphi(v_{i1}\odot\cdots \odot v_{ik_i})$ satisfies $e\varphi((v_{i1}\odot\cdots\odot v_{ij})\in s^{-1}P$ for some $j$. Hence $\varphi(\Gamma(u)\odot\Gamma(v_0)\odot\cdots\odot \Gamma(v_{r-1})\odot v_{r1}\odot\cdots\odot v_{rj})=s\cdot e\cdot \varphi(v_{r1}\odot\cdots\odot v_{rj})\in P$. Hence $\alpha$ has infinitely many prefixes in $T$, so $\alpha\in\lim(L)$.
\end{proof}

\subsection{Proof of Corollary \ref{cor:I-DBA-decidable}}
\begin{corollary*}
Let $K=\Gamma^{-1}(T)$ for some $T \in \rec(\mStar)$. Given $\aut{A}_K$, it is decidable in time $\bigO(|Q|^2\cdot |\Sigma|(|\Sigma|+\log|Q|))$ whether or not $K$ is limit-stable.
\end{corollary*}
\begin{proof}
Let $\aut{A}_K=(Q,\Sigma,q_0,\delta,F)$. Write $\aut{A}_{q,q'}=(Q,\Sigma,q,\delta,\set{q'})$ and $\aut{A}_q=(Q,\Sigma,q,\delta,F)$. Denote by $L_E(\aut{A}_K)$ the language recognized by $\aut{A}_K$ as an $E$-automaton (reachability condition). Note that $F,I$-cycle closure is equivalent to the following property: For every state $q\in Q$ the language $K_q=L(\aut{A}_{q,q})\cap L_E(\aut{A}_q)$ is trace-closed. 

Since $K_q$ is regular and a DFA for $K_q$ can be constructed from $\aut{A}_K$ in $\bigO(|Q|\cdot |\Sigma|)$ (take $Q\times\set{0,1}$ as states and memorize reaching $F$ in the second component), we can obtain the minimal DFA for $K_q$ from $\aut{A}_K$ in time $\bigO(|Q|\cdot |\Sigma|+|Q|\cdot|\Sigma|\cdot\log|Q|)=\bigO(|Q|\cdot|\Sigma|\cdot\log|Q|)$ using Hopcroft's algorithm. We then have to check if this automaton is $I$-diamond. This requires time $\bigO(|Q|\cdot |\Sigma|^2)$. So we have time $\bigO(|Q|\cdot |\Sigma|(|\Sigma|+\log|Q|))$ for every $q\in Q$.
\end{proof}

\end{document}